\documentclass[conference]{IEEEtran}
\IEEEoverridecommandlockouts
\usepackage{time}
\usepackage{cite}
\usepackage{amsmath,amssymb,amsfonts}
\usepackage{algorithmic}
\usepackage{comment}
\usepackage{graphicx}
\usepackage{textcomp}
\usepackage{xcolor}
\usepackage{amsthm}
\usepackage{lipsum}
\usepackage{tikz}
\usepackage[ruled,vlined,linesnumbered]{algorithm2e}

\def\BibTeX{{\rm B\kern-.05em{\sc i\kern-.025em b}\kern-.08em
    T\kern-.1667em\lower.7ex\hbox{E}\kern-.125emX}}

\newtheorem{theorem}{Theorem}
\newtheorem{lemma}{Lemma}
\newtheorem{definition}{Definition}
\newtheorem{problem}{Problem}
\newtheorem{corollary}{Corollary}
\newtheorem{proposition}{Proposition}
\newtheorem{remark}{Remark}

\newcommand*{\QEDB}{ } 
\newcommand{\RR}{\mathbb{R}^+}
\def \NN {{\mathbb{N}}}
\newcommand{\ZZ}{\mathbb{Z}}
\newcommand{\TA}{\mathcal{A}}

\newcommand{\taa}{{\em timed automaton}}

\newcommand{\LTS}{\mathcal{L}}

\newcommand{\SVM}{{\bf SVM}}

\newcommand{\CF}{{\tt cf}}

\newcommand{\EUP}{{\tt eup}}
\newcommand{\ELB}{{\tt elb}}
\newcommand{\PR}{{\tt \Theta}}
\newcommand{\pr}{{\tt \theta}}
\newcommand{\pcc}{{\em parametrically constrained clock}}

\newcommand{\Model}[1]{[\![#1]\!]}
\newcommand{\deff}{\  \widehat{=}\ } 

\newcommand{\rev}[1]{#1}
\newcommand{\ly}[1]{#1}

\graphicspath{{img/}}

\begin{document}

\title{Parameter Synthesis Problems for  \\ Parametric Timed Automata}

\author{\IEEEauthorblockN{Liyun Dai\quad  Bo Liu\quad Zhiming Liu}
\IEEEauthorblockA{\textit{RISE \&  School of Computer and Information Science, }\\
Southwest University, Chongqing, China \\
$\{$dailiyun,liubocq,zhimingliu88$\}$@swu.edu.cn}
\and
\IEEEauthorblockN{Taolue Chen}
\IEEEauthorblockA{\textit{Department of Computer Science and Information Systems} \\
	\textit{Birkbeck, University of London}\\
taolue@dcs.bbk.ac.uk}
}

\maketitle

 
\begin{abstract}	 
We consider the  {\em parameter synthesis problem} of parametric timed automata (PTAs). The problem is,   given a PTA and a  property, to compute the set of valuations of the parameters under which the resulting timed automaton satisfies the property. Such a set of parameter  valuations is called a {\em feasible region} for the PTA and the property. The problem is known undecidable in general. This paper, however, presents our study on some  decidable sub-classes  of PTAs and proposes efficient parameter synthesis algorithms for them.
Our contribution is  four-fold:  i) the study of the PTAs (called one-one PTAs) with one  parameter and one parametrically constrained clock and an algorithm for computing  the feasible region for a one-one PTA and a property; ii) the study of the PTAs  with lower-bound or with upper-bound parameters only and   a procedure to  construct the feasible region for such a PTA with  one parametrically
constrained clock  and a property;  iii) a theorem showing  that the feasible region for a PTA with  both the lower-bound and upper-bound parameters (i.e. a general L/U PTA) and a property which has existential quantifiers only  is a ``single connected" set; and iv) a support vector machine  based algorithm to identify the boundary  of the  feasible region for a general L/U  PTA and a property. We belief that these results contribute to advancing the theoretical investigations of the parameter synthesis problem for PTAs, and support to exploit machine learning methods to give potentially  more practical synthesis algorithms as well.
\end{abstract}

%
%
%

\section{Introduction}
\label{sec:inr}
Real-time applications are  increasing importance, so are  their complexity and requirements for trustworthiness,  in the era of Internet of Things (IoT), especially  
in the areas of industrial control and smart homes. Consider, for example, the control system of a boiler used in house. Such a system is required to  switch  on the gas within a certain bounded period of time  when the water gets too cold.  Indeed, the design and implementation of the system  not only have to  guarantee the correctness of system functionalities,  but also need to assure  that the application is in compliance with the non-functional requirements, that are timing constraints  in this case. 

{\em Timed automata} (TAs)  \cite{Alur90,alur1994a} are widely used  for modeling and verification of real-time systems. However, one disadvantage of the TA-based approach is that  it  can only be used to  verify \emph{concrete} properties, i.e., properties with concrete values of all timing parameters occurring in the system. Typical examples of such  parameters  are  upper and lower bounds of computation time, message delay and time-out. This makes the traditional TA-based approach  not ideal for the design of real-time applications because in the \emph{design phase}  concrete values are often  not available. This problem is usually dealt with extensive trial-and-error and prototyping activities to find out what concrete values of the parameters are suitable. This  approach of design is costly, laborious, and error-prone, for at least two reasons:  
(1) many trials with different parameter configurations suffer from unaffordable costs, without enough assurance  of a safety standard because a sufficient coverage of configurations is  difficult to achieve; (2) little  or no feedback information is provided to the developers to help improve the design when  a system malfunction is detected. 

\subsection{Decidable parametric timed automata}
To mitigate the limitations of the TA-based approach, 
{\em parametric timed automata} (PTAs) are  proposed \cite{alur1993parametric,annichini2000symbolic,bandini2001application,HUNE2002183}, which  allow more general constraints on  invariants of nodes (or states) and guards of edges (or transitions) of an automaton. Informally, a  clock $x$ of a PTA $\TA$  is called a \pcc\ if $x$ and some parameters both occur in a constraint  of $\TA$.    Obviously, given any  valuation of the  parameters in a PTA, we obtain a concrete TA. One of the most important questions of PTAs is the \emph{synthesis problem}, that  is, for a given property to compute the entire set of valuations of the parameters for a PTA such that when the parameters are instantiated by these valuations, the resulting TAs all satisfy the property. The  synthesis problem for general PTAs is known to be undecidable. There are, however, several proposals  to restrict the general PTAs from different perspectives to gain decidability.   
Two kinds of restrictions that are being widely investigated are (1) on  the number of clocks/parameters in the PTA;  and (2) on the way in which  parameters are bounded, such as the  L/U PTAs \cite{HUNE2002183}. 

\subsection{Our contribution}
The first part of our work  is about  restrictions of the first kind above, and  it considers the PTAs, which we later refer to as  {\em one-one PTAs},  which have  {\em one parametrically constrained clock} and one parameter, but allowing arbitrary number of other clocks. We extend the result of \cite{alur1993parametric} and  provide an algorithm to  construct  the feasible parameter region explicitly for a  one-one PTA and a property.

The second part of our work studies {\em L/U automata}. In an L/U automaton, each parameter occurs either as a lower-bound only in the invariants and guards, or as an upper-bound only therein. In other words, a parameter in an L/U automaton cannot occur  as both a lower-bound \emph{and} an upper-bound of clocks. We call an L/U automaton an {\em L-automaton} (resp. {\em U-automaton}) if all parameters occur only as lower-bounds (resp. upper-bounds).  The results of \cite{HUNE2002183} show    that the emptiness problem for L/U automata is decidable. They also extend the model checker {\tt UPPAAL} to synthesize linear parameter constraints for L/U-automata. Decidability results for L/U automata have been further investigated \cite{bozzelli2009decision}. There for  L-automata and U-automata, the authors solve the synthesis problem for a restricted class of liveness  properties, i.e. the  existence of an \emph{infinite accepting run} for the automaton. Our work in this paper, instead of   the liveness property  considered in \cite{bozzelli2009decision}, considers  an other class of properties. These properties  are generally generally described  as formulas in  temporal logic of the form $\exists \Diamond \phi$ and $\forall \Box \phi$, and their satisfiction  by a PTA can be  treated as reachability properties. Here, $\phi$ is a state property, $\exists$ (or $\forall$) means there exists (resp. for all)  runs.  For these properties, we solve the parameter synthesis problem for  L-automata and U-automata by   explicitly constructing the feasible parameter regions. 

Furthermore, for  general model of L/U automata, we show that the feasible parameter region forms a ``single connected"  set provided that the property contains existential quantifiers only. Being  connected here means that for any pair of valuations $v$ and $v'$ there is at least one sequence of $v=v_1,\ldots, v_\ell = v'$ feasible valuations, such that the Euclidean distance between $v_i$ and $v_{i+1}$ is $1$.
This  topological property of feasible regions allows us to develop a machine learning algorithm based on support vector machine (\SVM) to identify the  boundary of a feasible region. 
%

 \subsection{Related work} 
 The earliest work on PTAs goes back to 90's  \cite{alur1993parametric}  by Alur, {\em et al}, where   the general undecidability  of the reachability emptiness for a PTA with three or more parametrically constrained clocks  is proved. There, a  backward computation based algorithm to solve the {\em emptiness problem}  is also presented   for a nontrivial class of PTAs  which have only one parametrically constrained clock.  It is  also shown there    that for the remaining class of PTAs, that is the class of PTAs with exactly two parametrically constrained clocks, 
 the problem is closely related to various hard (viz. open) problems of logic and automata theory. A semi-algorithm based on expressive symbolic representation structures called parametric difference bound matrices  is   proposed in \cite{annichini2000symbolic}. The algorithm  uses accurate extrapolation techniques to speed up the reachability computation and ensure termination.
 The work in \cite{alur2001parametric} proposes a class of PTAs in which a parameter cannot be shared by a lower bound constraint and an upper bound constraint.   And, in this setting, the work there   studies  the Linear Temporal Logic (LTL) augmented with parameters. 
 
  \rev{A SMT-based method of computing under-approximation
  	of the solution to this problem for L/U automata is provided in \cite{KP12a}. \cite{AL17} further studies L/U automata by considering liveness related problems.} 

  Symbolic algorithms are proposed in  \cite{jovanovic2015integer} to synthesize all the values of parameters for the reachability and unavoidability properties for bounded integer-valued parameters. A  proof is given  in  \cite{bundala2014advances} to the decidability of the emptiness problem for  the class of PTAs  which have  two parametrically constrained clocks  and one parameter.  An  adaption of  the counterexample guided abstraction refinement (CEGAR)  is used in  \cite{frehse2008counterexample} to  obtain an under-approximation of the set of good parameters using linear programming.  An method, called an ``inverse method'', is  provided  in \cite{andre2009inverse}. This method,  for a given set of sample parameter valuations as the input,  synthesizes a constraint on  the parameters such that i) all  sample valuations satisfy the constraint and, ii) the TAs defined by any two parameter valuations satisfying the constraint are time-abstract equivalent.
The work in  \cite{andre2015language} considers the class of  deterministic PTAs with a single lower-bound integer-valued parameter or a single integer-valued upper-bound parameter and one extra (unconstrained) parameter.  There, it also shows that, for these PTAs, the language-preservation problem  is proved to be  decidable. The  PTAs that we consider  in this paper  is orthogonal to those which  are  presented in \cite{andre2015language}.  \rev{Instead of  synthesizing the full set of
	parameter constraints in general, \cite{Knapik2012BMC} presents method to obtain a part of this set. The work in \cite{benes2015language} considers the class of   emptiness problem of PTA with one  parametrically constrained clock. Our idea is similar with its', both given an upper bound of parameter and prove that when the value of parameter greater than this upper 
bound, the behaviour of corresponding timed automata are same under abstract view.  A machine learning based method for synthesizing constraints on the parameters, which guarantee the system behaves according to certain properties, is provided in \cite{LSGA17}.} 
Finally, we refer to  \cite{Andr16} for a survey of recent progress in decidability problems of PTAs. 

 \subsection{Organization}
 We define  in Section \ref{sec:pre}
the  model of PTA  and present the  relevant notions and properties of PTAs. In  Section \ref{sec:oneP}, we study one-one PTAs and present  the algorithm to compute the feasible region for  a one-one PTA and a property. We present, in Section \ref{sec:newresult}, the work on the parameter synthesis problem for L/U PTAs. We prove in Section~\ref{sec:furth}  the theorem of strong connectivity of feasible regions for general L/U automata,  and based on this theorem we  present a  learning-based method to identify the boundary  of  a feasible parameter region. Finally, we draw the conclusions in Section \ref{sec:con}.


\section{Parametric Timed Automata}
\label{sec:pre}
We introduce the basis of  PTAs  and set up terminology for our discussion. We first define some preliminary notations before we introduce PTAs. We will use a model of  labeled transition systems (LTS) to define  semantic behavior of PTAs.

\subsection{Preliminaries}

We use $\ZZ$, $\NN$, $\mathbb{R}$ and $\RR$ to denote the  sets of integers, natural numbers, real numbers and non-negative real numbers, respectively.  Although each  PTA involves only a finite number of clocks and a finite number parameters, we need an infinite set  of {\em clock variables} (also simply called {\em clocks}), denoted by $\mathcal{X}$ and an infinite set of {\em parameters}, denoted by $\mathcal{P}$, both are enumerable. We use $X$ and $P$ to denote (finite) sets of clocks and parameters and   $x$ and $p$, with subscripts if necessary,  to denote clocks and parameters, respectively. 

We mainly consider dense time, and thus we  define a {\em clock valuation} $\omega$
as a function of the type $\mathcal{X}\mapsto \RR$ from the set of clocks to the  set of non-negative real numbers,   assigning each clock variable a nonnegative real number. For a finite set $X=\{x_1, \ldots, x_n\}$ of clocks, an evaluation $\omega$  restricted on $X$ can be represented by  a $n$-dimensional   point $\omega(X)=(\omega(x_1),
\omega(x_2),\ldots,\omega(x_n))$, and it is called an {\em parameter valuation of $X$} and simply denoted as $\omega$ when there is no confusion.  Similarly,  a {\em parameter valuation} $\gamma$  is an assignment of values to the  parameters, but the values are natural numbers, that is $v: \mathcal{P}\mapsto \NN$.  For a finite set $P=\{p_1,\ldots, p_m\}$ of $m$ parameters, a parameter valuation $\gamma$ restricted on $P$ corresponds to  a $m$-dimensional point $(\gamma(p_1),\gamma(p_2),\ldots,\gamma(p_m))\in \NN^m$, and we use this vector to denote the valuation $\gamma$ of $P$ when there is no confusion. 
		
\begin{definition}[Linear expression]
A linear expression $e$ is an expression of the form $c_0+ c_1p_1+\cdots+c_mp_m$,  where $c_0,\cdots,c_n\in \ZZ$.   
\end{definition}
We use $\mathcal{E}$ to denote the set of linear expressions, $\textit{con}(e)$  the constant $c_0$, and  $\CF(e,p)$  the coefficient of $p$ in $e$, i.e. $c_i$ if $p$ is $p_i$ for $i=1,\ldots, m$, and $0$,  otherwise. For the convenience of discussion, we also say the infinity $\infty$ is a linear expression.  We call  expression $e$ a {\em parametric expression} if $c_i\neq 0$ for some $i\in \{1,\ldots, m\}$, a {\em concrete  expression}, otherwise (i.e., $e$ is parameter free).

A PTA only allows {\em parametric constraints} of the form  $x-y\sim e$, where $x$ and $y$ are clocks, $e$ is a linear expression, and the  ordering  relation ${\sim}\in  \{>,\ge, <,\le, =\}$. A constraint $g$ is called a  {\em parameter-free} (or {\em concrete}) {\em  constraint} if the expression in it is concrete.  For a linear expression $e$, a parameter valuation $\gamma$, a clock valuation $\omega$ and a constraint $g$, let 
\begin{itemize}
\item $e[\gamma]$ be the (concretized) expression obtained from $e$ by substituting the value $\gamma(p_i)$ for $p_i$ in $e$, i.e. $c_0+ c_1\times \gamma(p_1)+ \ldots + c_m\times \gamma(p_m)$,
\item $g[\gamma]$  be the predicate obtained from constraint $g$  by substituting the value $\gamma(p_i)$ for $p_i$ in $g$, and 
\item $\omega \models g$ holds if $g[\omega]$  holds.
\end{itemize}

A pair $(\gamma, \omega )$ of parameter valuation and clock valuation   gives an evaluation to any parametric constraint $g$. We use $g[\gamma, \omega]$ to denote the truth value  of $g$ obtained by substituting each parameter $p$  and each clock $x$  by their values $\gamma(p)$ and $\omega(x)$, respectively. We say the pair of valuations $(\gamma, \omega)$ satisfies constraint $g$, denoted by $(\gamma, \omega)\models g$,  if $g[ \gamma, \omega]$ is evaluated to true. For a given parameter valuation $\gamma$, we define  $\Model{g[\gamma]}= \{\omega\mid (\gamma, \omega)\models g\}$ to be the set of clock valuations which together with $\gamma$ satisfy $g$.

A clock $x$ is reset by an {\em update} which is an expression of the form $x:= b$, where  $b\in \NN$. Any reset $x:=b$ will change a clock valuation $\omega$ to a clock valuation $\omega'$ such that $\omega'(x) =b$ and $\omega'(y) =\omega(y)$  for any other clock $y$. Given a clock valuation $\omega$ and a set $u$ of updates, called an {\em update set}, which contains at most one reset for one  clock,  we use $\omega[u]$ to denote the clock valuation after  applying all the clock resets in $u$ to $\omega$.  We  use $c[u]$ to denote the constraint which is used to assert the relation of the parameters with the clocks values after the clock resets of $u$. Formally,  $c[u](\omega)\deff c(\omega[u])$ for every clock valuation $\omega$. 

It is easy to see that the general constraints $x-y\sim e$ can be expressed in terms of {\em atomic constraints} of the form $b_1x-b_2y\prec e$,  where ${\prec} \in \{<,\le\}$ and $b_1,b_2\in \{0,1\}$. To be explicit, an atomic constraint is in one of the following three  forms $x-y \prec e$, $x\prec e$, or $-x\prec e$.  We can write  $-x_i \prec  e$ as $x_i\succ -e$ and $x-y\prec e$ as $y-x\succ -e$, where ${\succ} \in \{>,\geq\}$.  However, in this paper we mainly  consider  {\em simple constraints}  that are  finite  conjunctions of atomic constraints. 

\subsection{Parametric timed automata}
We assume the knowledge of timed automata (TAs), e.g., \cite{alur1999timed,bengtsson2004timed}. A  clock constraint of a TA  either a  {\em invariant property} when the TA is  in a state (or location)  or a {\em guard condition} to enable the changes of states (or a state transition). Such a constraint  is  in general  a Boolean expression of parametric free  atomic constraints. However, we can assume that the  guards and invariants of TA are simple concrete  constraints, i.e. conjunctions of concrete atomic constraints. This is because we can always transform a TA  with  disjunctive guards and invariants to an equivalent  TA  with  guards and invariants which are simple constraints only.   

In what follows, we define  PTAs which extend TAs  to allow the use of parametric simple constraints as guards and invariants (see \cite{alur1993parametric}).

\begin{definition}[PTA]
	\label{def:pta}
	Given a finite set of clocks $X$ and a finite set of parameters $P$, a PTA   is a 5-tuple $\TA=(\Sigma,Q,q_0,I,\rightarrow)$,
	where
	\begin{itemize}
		\item $\Sigma$ is a finite set of actions.
		\item  $Q$ is a finite set of locations and  $q_0\in Q$ is  called the initial location,
		\item  $I$ is the invariant, assigning to every $q\in Q$ a simple constraint $I_q$ over 
		the clocks $X$ and parameters $P$, and
		\item $\rightarrow$
		is a discrete transition relation whose elements are of the form $(q,g,a,u,q')$, where $q,q'\in Q$, $u$ is an update set, $a\in \Sigma$ and $g$  is a simple constraint.
	\end{itemize}
\end{definition}
Given a PTA $\mathcal{A}$, a tuple $(q,g,a,u,q')\in {\rightarrow}$ is also denoted  by $q \xrightarrow{g\&a[u]} q'$, and it is called a transition step (by the guarded action $g\&a$). In this step, $a$ is the  action that triggers the transition. The constraint $g$ in the transition step is called the {\em guard} of the transition step, and only when $g$ holds in  a location  can the transition  take place. By this transition step, the system modeled by the automaton changes from location $q$ to location $q'$, and the clocks are reset by the updates in $u$. However, the meaning of the guards and clock resets and acceptable runs of a PTA will be defined by a labeled transition system (LTS) later on. At this moment, we define a {\em syntactic run} of a PTA $\mathcal{A}$ as a sequence of consecutive transitions step starting from the initial location
\[\tau = (q_0, I_{q_0})\xrightarrow{g_1\&a_1[u_1]}(q_1, I_{q_1})\cdots \xrightarrow{g_\ell \&a_\ell [u_\ell]}(q_\ell, I_{q_\ell})
\]
 
Given a PTA $\TA$, a clock $x$ is said to be a {\em parametrically constrained  clock} in $\TA$  if there is a parametric constraint containing $x$. Otherwise,   $x$ is a concretely   constrained  clock. We can follow the procedures in   \cite{alur1993parametric} and \cite{bundala2014advances} to  eliminate  from  $\TA$ all the concretely constrained clocks.  Thus, the  rest of this paper only considers the PTAs  in which all clocks are parametrically constrained. We use $\textit{expr}(\mathcal{A})$ and $\textit{para}(\mathcal{A})$ to denote the set of all linear expressions and parameters in a PTA $\mathcal{A}$, respectively.
 
\begin{figure}[h!]
	\centering
	\includegraphics[width=0.35\textwidth]{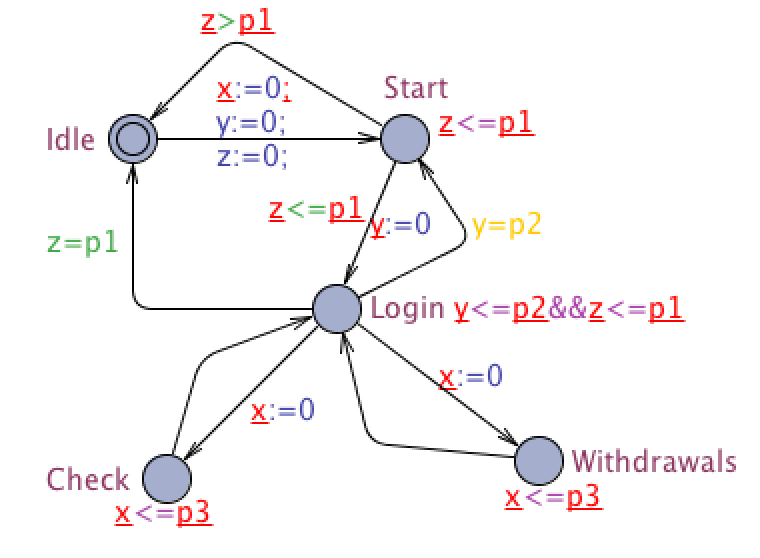}
	\caption{An ATM modeled using a PTA.} \label{fig:atm}
\end{figure}

\subsubsection*{Example 1}
The PTA in {\em Fig. \ref{fig:atm}} models an ATM. It has 5 locations, 3 clocks $\{x,y, z\}$  and 3 parameters $\{p_1,p_2, p_3\}$. This PTA is deterministic and
all the clocks are parametric. To understand the behavior of state transitions,  for examples, the  machine can initially idle for an arbitrarily long time. Then,  the user can start the system by, say, pressing a button and the PTA enters location ``Start" and resets the three clocks. The machine can remain in ``Start" location as long as the invariant  $z\le p_1$ holds, and during this  \ly{time  the user can drive the system} (by pressing a corresponding button) to login their account and the automaton  enters  location ``Login" and resets clock $y$.  A time-out action occurs and it  goes back to ``Idle'' if  the machine stays at ``Start''  for too long and the  invariant $z\leq p_1$ becomes false. Similarly, the machine can remain in location ``Login" as long as the invariant $y\le p_2 \wedge z\le p_1$ holds and during this time the user can decide either to ``Check'' (her balance) or to ``Withdraw" (money), say by pressing  corresponding buttons.   However,  if the user does not take any of these actions $p_2$ time units after the machine enter location ``Login",  the machine will back  to ``Start" location.

\subsection{Semantics of PTA via labeled transition systems}
We use a  standard model of  {\em labeled transition systems} (LTS) for describing  and analyzing the behavioral properties of  PTA.
\begin{definition}[LTS]
	\label{def:lts} A \emph{labeled transition system} (LTS) over a set of (action) symbols $\Delta$ is a triple
	$\LTS=(S,S_0,\rightarrow)$, where
	\begin{itemize}
	\item  $S$ is a set of states with a subset  $S_0\subseteq S$ of states called the  initial states.
	\item ${\rightarrow} \subseteq S\times \Delta  \times S$  is a relation, called the transition relation. 
	\end{itemize}
	We write $s \xrightarrow{a} s' $ for a triple $(s,a,s') \in {\rightarrow}$ and it is called a transition step by action $a$. 
	
	A run  of $\LTS$ is a finite alternating sequence of states in $ S$ and actions  $\Delta$, $ \xi = s_0a_1s_1\ldots a_\ell  s_\ell $,   such that $s_0\in S_0$ and $s_{i-1}\xrightarrow{a_i}s_i\in \rightarrow$  for $i=1,\ldots,\ell$. A run $\xi$ can be written in the form of $s_0\xrightarrow{a_1} s_1 \xrightarrow{a_2} \cdots \xrightarrow{a_{\ell}} s_\ell $. 
	
	The length of a run $\xi$ is its number $\ell$ of transitions steps and it is denoted as $|\xi|$, and a state $s\in S$ is called reachable in $\LTS$ if $s$ is the last state a run of $\LTS$, e.g. $s_\ell$ of $\xi$.  
\end{definition}

\begin{definition}[LTS semantics of PTA]
	\label{def:sem}
	For a PTA $\TA =(\Sigma, Q, q_0, I,\rightarrow)$ and a parameter valuation $\gamma$, the concrete semantics of PTA under $\gamma$, denoted by $\TA[\gamma]$,  is the LTS $(S,S_0,\rightarrow)$ over $\Sigma\cup \RR$, where
\begin{itemize}
\item a state in  $S$  is a  location $q$ of  $\mathcal{A}$ augmented with the clock valuations which together with the parameter valuation $\gamma$ satisfy the invariant   $I_q$ of the location, that is
	\[S=\{ (q,\omega)\in Q\times (X\rightarrow \RR)\mid (\gamma, \omega)\models I_q\}\]
  \[	S_0=\{(q_0,\omega)\mid (\gamma, \omega) \models I_{q_0} \wedge \omega=(0,\cdots,0)\} \]

\item  any transition step in the transition $\rightarrow$ of the LTS is either an instantaneous  transition step  by  an action in $\Sigma$ defined by $\mathcal{A}$ or by a time advance, that are specified by the following rules, respectively
	\begin{itemize}
		\item {\bf instantaneous transition}:  for any $a\in \Sigma$, $(q,\omega)\xrightarrow{a}(q',\omega') $ if  there are  simple constraint $g$ and an  update set  $u$ such that $q\xrightarrow{ g\&a[u]} q'$, $(\gamma,\omega) \models g $ and $\omega'=\omega[u]$; and 
		\item {\bf time advance transition} $(q,\omega) \xrightarrow{d} (q',\omega') $ if $q'=q$ and $\omega'=\omega+d$.
	\end{itemize}
  \end{itemize}
\end{definition}
A {\em concrete run} of a PTA $\mathcal{A}$ for a given valuation $\gamma$ is a sequence of consecutive state transition steps $\xi=s_0\xrightarrow{t_1} s_1 \xrightarrow{t_2} \cdots \xrightarrow{t_{\ell}} s_\ell $  of the LTS $\mathcal{A}[\gamma]$, which we also call a run of the LTS $\mathcal{A}[\gamma]$.  A state $s = (q,\omega)$ of $\mathcal{A}[\gamma]$ is a {\em reachable state} of  $\TA[\gamma]$  if there exists some run  $\xi=s_0\xrightarrow{t_1} s_1 \xrightarrow{t_2} \cdots \xrightarrow{t_{\ell}} s_\ell $ of  $\TA[\gamma]$ such that $s=s_\ell $.

Without the loss of generality,
we merge  any two consecutive time advance transitions respectively labelled by  $d_i$ $d_{i+1}$ into a single time advance transition labels by $d_i+d_{i+1}$. We can further merger a consecutive pair $s\xrightarrow{d} s'\xrightarrow{a} s''$ of a timed advance transition by $d$ and  an instantaneous transition by an action $a$  in a run into a single observable transition  step $s\xrightarrow{a} s''$. If we do this repeatedly until all time advance steps are eliminated, we obtain an {\em untimed run} of the PTA (and the LTS), and the sequence of actions in an untimed run is called a  {\em trace}.

We call an untimed run  $\xi=s_0\xrightarrow{a_1}s_1\cdots \xrightarrow{a_\ell}s_\ell$  a {\em simple run}  if $\omega_i\ge \omega_{i-1}$ for  $i=1,\cdots, \ell$, where $s_i=(q_i,\omega_i)$.
It is easy to see that $\xi$ is a {\em simple untimed run} if each transition by $a_i$ does not have any clock reset in $\xi$. 

\begin{definition}[LTS of trace]
 For a   PTA  $\mathcal{A}$ and a syntactic run
 \[\tau=(q_0{,}I_{q_0})\xrightarrow{g_1\&a_1[u_1]}(q_1{,}I_{q_1}){\cdots} \xrightarrow{g_\ell{\&}a_\ell[u_\ell]}(q_\ell{,}I_{q_\ell})\]
we define the PTA $\TA_\tau{=}(\Sigma_\tau {,}Q_\tau {,}q_{0,\tau}{,}I_\tau,\rightarrow_\tau)$, where 
 \begin{itemize}
 \item $\Sigma_\tau =\{a_i \mid i= 1,\cdots, \ell\}$,
 \item $Q_\tau =\{q_0,\cdots, q_\ell\}$ and $q_{0,\tau}=q_0$,
 \item  $I_\tau (i)=I_{q_i}$ for $i\in Q$, and 
 \item $\rightarrow_\tau =\{(q_{i-1},g_i, a_i, u_i,q_i)\mid i=1,\cdots, \ell \}$.
 \end{itemize}
Give a parameter valuation $\gamma$, the  concrete semantics
of $\tau$ under $\gamma$ is  defined to be the LTS $\TA_\tau[\gamma]$. 
\end{definition}
For a syntactic run 
 \[\tau=(q_0{,}I_{q_0})\xrightarrow{g_1\&a_1[u_1]}(q_1{,}I_{q_1}){\cdots} \xrightarrow{g_\ell{\&}a_\ell[u_\ell]}(q_\ell{,}I_{q_\ell})\]
We use $R(\mathcal{A}_\tau[\gamma])$ to denote the set of  states $(q_k, \omega_k)$ of  $\mathcal{A}_\tau[\gamma]$ such that the following is a untimed run of $\TA_\tau[\gamma]$
 \[(q_0,\omega_0)\xrightarrow{a_1} (q_1,\omega_1) \cdots \xrightarrow{a_k}(q_k, \omega_k) \cdots\xrightarrow{a_\ell}(q_\ell, \omega_\ell)\]

\subsection{Two decision problems for PTA}
We first present the  properties of PTAs  which we consider in this paper. 

\begin{definition}[Properties]
	\label{def:prop}
	A   state  and a   property for a PTA  are specified by a state predicate $\phi$ and a temporal formula $\psi$ defined by the following syntax, respectively: for $x,y\in X$, $e\in \mathcal{E}$ and $\prec\in \{<,\le, =\}$  and $q$ is a location.
	
	\[\begin{array}{lllll}
	  \phi&::=& x\prec e \mid -x \prec e\mid x-y\prec e\mid \\
	  && q\mid \neg \phi\mid \phi  \wedge \phi \mid \phi \vee \phi\\
	  \psi&::=&\forall \Box\phi \mid \exists \Diamond \phi
	  \end{array}
	  \]

\end{definition}

Let  $\gamma$ be a parameter valuation  and   $\phi$ be  state formula. We  say $\TA[\gamma]$ {\em satisfies} $\exists\Diamond \phi$, denoted  by $\TA[\gamma]\models \exists\Diamond \phi$,  if there  is a  reachable state $s$ of  $\TA[\gamma]$  such that $\phi$ holds in  state $s$. We call these properties 
{\em reachability  properties}.
 Similarly,  $\TA[\gamma]$ {\em satisfies}  $\forall \Box \phi$, denoted by $\TA[\gamma] \models \forall \Box \phi$,  if $\phi$ holds in all reachable states of $\TA[\gamma]$.  We call these properties {\em safety properties}. We can see that if $\TA[\gamma]{\models}\exists\Diamond \phi$,  there is an syntactic  run $\tau$ such that there is  a state in  $R(\mathcal{A}_\tau[\gamma])$  satisfies  $\phi$.
In this case, we also say that the syntactic   run $\tau$ satisfies $\phi$ under the parameter valuation $\gamma$. We denote it by $\tau[\gamma]\models \phi$.

We are now ready to present the formal statement of the parameter synthesis problem and the emptiness problem of PTA.

\begin{problem}[The parameter synthesis problem]
	\label{pro:synth}
	Given a PTA $\TA$ and a property  $\psi$, compute the entire set  $\Gamma(\TA,\psi)$ of parameter valuations such that $\TA[\gamma]\models \psi$ for each  $\gamma \in \Gamma(\TA,\psi)$.
\end{problem}

 Solutions to the problems are important in system plan and optimization design.
Notice that when there are no parameters in $\mathcal{A}$,  the problem is decidable in PSPACE \cite{alur1994a}. This  implies  that if there are parameters in $\mathcal{A}$, the satisfaction  problem $\TA[\gamma]\models \psi$ is decidable in PSPACE for any given parameter valuation $\gamma$.

A special case of the synthesis problem is the emptiness problem, which is  by itself very important and formulated below.
\begin{problem}[Emptiness problem]
	\label{pro:emp}
	Given a PTA  $\TA$ and a  property $\psi$, is there a parameter valuation $\gamma$  so that $\TA[\gamma]\models \psi$?
\end{problem}
This is equivalent to the problem of checking if the set $\Gamma(\mathcal{A}, \psi)$ of feasible parameter valuations is empty.

Many safety verification problems can be reduced to the emptiness problem. We say that {\em Problem \ref{pro:emp}} is a special case of  {\em Problem \ref{pro:synth}} because solving the latter for a PTA $\TA$ and a property $\psi$ solves {\em Problem \ref{pro:emp}}. 

It is known that  the emptiness problem is decidable for a PTA with only one clock \cite{alur1993parametric}. However, the problem becomes undecidable for PTAs with more than two clocks~\cite{alur1993parametric}. 
Significant progress could only be made in 2002  when the  subclass of L/U PTA  were  proposed in \cite{HUNE2002183} and the emptiness problem was proved to be decidable for these automata. In the following, we will extend these results and define some classes of PTAs for which  we propose solutions to the parameter synthesis problem and the emptiness problem.


\section{Parameter  synthesis of PTA with one parametric clock}
\label{sec:oneP}
In this section, we present our first contribution and the solution to the   {\em parameter synthesis problem} of PTAs  with one parametric clock $x$ and one parameter which we call   {\em one-one PTAs}. A  one-one PTA allows an arbitrarily number of  {\em concretely constrained clocks}, and we denote a one-one PTA $\mathcal{A}$ with the parametric  clock $x$ and parameter $p$ as $\mathcal{A}[x,p]$, and as $\mathcal{A}$ when there is no confusion. We use $\mathcal{A}[p=v]$ for the LTS (and the concrete PTA) under the valuation $p=v$. Our main theorem is that the entire set of feasible parameter valuations $\Gamma(\mathcal{A}[x,p], \psi)$  is computable for any one-one PTA $\mathcal{A}[x,p]$ and any  property $\psi$ defined. The theorem is formally stated below.

\begin{theorem}[Synthesisability of one-one PTA]
\label{th:one-one-PTA}
	 The set $\Gamma(\TA,\psi)$  of feasible parameter valuations is solvable  for any one-one PTA $\mathcal{A}[x,p]$ and any  property $\psi$.
\end{theorem}
The establishment and proof of this theorem involve a sequence of techniques to reduce the problem to computing the set of reachable states of an LTS. The major steps of  reduction include 
\begin{enumerate}
\item Reduce the problem of satisfaction of a property $\psi$, say in the form of $\exists \Diamond \phi$,  by a  syntactic run $\tau$ to a reachability problem. This is done by encoding the state property in $\psi$ as a conjunction of  the invariant of a state.
\item Then we   move the  state invariants in a syntactic run out of the states and conjoin them to the guards of the corresponding transitions. 
\item Construct feasible runs for a given syntactic run in order to reach a given location.  This requires to define the notions of effect lower and  effect upper bounds of guards of transitions, through which an lower bound of feasible  parameter valuation is defined. 
\end{enumerate} 

 \subsection{Reduce satisfaction of  system to  reachability problem} 
We note  that $\psi$ is either of the form $\exists \Diamond \phi$ or  the  dual form $\forall  \Box \phi$, where $\phi$ is a state property.  Therefore, we only need to consider the problem of computing the set $\Gamma(\TA,\psi)$  for the case when $\psi$ is a formula of the form  $\exists \Diamond \phi$, i.e., there is a syntactic run $\tau$ such that  $\tau[\gamma] \models \phi$ for every $\gamma\in \Gamma(\TA,\psi)$. Our idea is to reduce the problem  of deciding $\TA \models \psi$ to a reachability problem of an LTS by encoding the state property $\phi$ in  $\exists \Diamond \phi$  into the guards of the transitions of $\mathcal{A}$. 
 
\begin{definition}[Encoding state property]
	Let  $\phi$ be a  state formula and $q$ be a location. We definite $\alpha(\phi,q)$ as follows, where $\equiv$ is used to denote syntactic equality between formulas:
	\begin{itemize}
		\item  $\alpha(\phi, q)\equiv \phi$ if $\phi \equiv x-y\prec e$,  $\phi \equiv x\prec e$ or   $\phi \equiv -x\prec e$,  where $x$ and $y$ are clocks and $e$ is an expression. 
		\item  when $\phi$ is a location $q'$,  $\alpha(\phi,q')\equiv true$ if $q'$ is $q$ and  $false$ otherwise.
		\item $\alpha $ preserves all Boolean connectives, that is   $\alpha(\neg \phi_1,q)\equiv \neg\alpha(\phi_1,q)$,  $\alpha(\phi_1\wedge\phi_2,q)\equiv \alpha(\phi_1,q)\wedge \alpha(\phi_2,q)$, and  $\alpha(\phi_1\vee \phi_2,q)\equiv  \alpha(\phi_1,q)\vee \alpha(\phi_2,q)$.
	\end{itemize}
\end{definition}
We can easily prove the   following lemma.

\begin{lemma}
	\label{lem:moveP}
	Given a PTA\ $\TA$, $\psi\equiv \exists\Diamond \phi$,  and  a syntactic run of $\mathcal{A}$ 
	\[\tau=(q_0,I_{q_0})\xrightarrow{g_1\& a_1[u_1]}(q_1,I_{q_1}){\cdots} \xrightarrow{g_\ell\&a_\ell[u_\ell]}(q_\ell,I_{q_\ell})\] 
we overload the function notation $\alpha$ and define the encoded run $\alpha(\tau)$ to be 
		\[
	 (q_0,I_{q_0})\xrightarrow{g_1\&a_1[u_1]}(q_1{,}I_{q_1}){\cdots}\xrightarrow{g_\ell\&a_\ell[u_\ell]}(q_\ell{,}I_{q_\ell }{\wedge} \alpha(\phi{,}q_\ell))
\]
	Then $\tau$ satisfies $\psi$ under parameter valuation $\gamma$ if and only if $R(\mathcal{A}_{\alpha (\tau)}[\gamma])\neq \emptyset$.
\end{lemma}
 Notice the  term guard is slightly abused in the lemma as $\alpha(\phi,q_\ell)$  may have disjunctions, and thus it may not be a simple constraint.
 
 \subsection{Moving state  invariants to guards of transitions}
  It is easy to see that both the invariant $I_q$  in the pre-state of the transition and the guard $g$ in a transition step  $(q,I_q)\xrightarrow{g\&a[u]} (q',I_{q'})$ are both enabling conditions for the transition to take place.  Furthermore, the invariant $I_{q'}$  in the post-state of a transition needs to be guaranteed by the set of clock resets $u$. Thus we can also understand this constraint as a guard condition for the transition to take place (the transition is not allowed to take place if the invariant of the post-state is false). 
  
 For a PTA $\TA$ and a syntactic  run
	 \[\tau=(q_0,I_{q_0})\xrightarrow{g_1\&a_1[u_1]}(q_1{,}I_{q_1}){\cdots}\xrightarrow{g_\ell \&a_\ell[u_\ell]}(q_\ell ,I_{q_\ell}).\] 
	 Let $\overline{g}_i = (g_i\wedge I_{q_{i-1}}\wedge I_{q_i}[u_i])$. We define $\beta(\tau)$ as \[
		\begin{split}
	(q_0,true)\xrightarrow{\overline{g}_1\&a_1[u_1]}(q_1,true)\cdots
	\xrightarrow{\overline{g}_\ell \&a_\ell[u_\ell]}(q_\ell, true)
	\end{split}
\]	
\begin{lemma}	\label{lem:moveI}
	For a PTA $\TA$, parameter valuation $\gamma$ and  a syntactic  run
	 \[\tau=(q_0,I_{q_0})\xrightarrow{g_1\&a_1[u_1]}(q_1{,}I_{q_1}){\cdots}\xrightarrow{g_\ell \&a_\ell[u_\ell]}(q_\ell ,I_{q_\ell})\] 
we have $(\gamma, (0,\cdots, 0))\models I_{q_0}$ and $R(\mathcal{A}_{\beta(\tau)}[\gamma])\neq \emptyset$ if and only if  $R(\mathcal{A}_\tau[\gamma])\neq \emptyset $.
\end{lemma}
\begin{proof}
Assume $(\gamma, x=0)\models I_{q_0}$ and $R(\mathcal{A}_{\beta(\tau)}[\gamma])\neq \emptyset$.
There is run $\xi $ of $\mathcal{A}_{\beta(\tau)}[\gamma]$  which is an alternating sequence of instantaneous and time advance transition steps 
  \[\xi= (q_0{,}\omega_0)\xrightarrow{d_0} (q_0{,}\omega_0') \xrightarrow{a_1} (q_1{,}\omega_1) \cdots   \xrightarrow{a_\ell} (q_\ell{,}\omega_\ell)\]
  such that $(\gamma, \omega_i')\models g_{a_{i+1}}\wedge I_{q_{i}}\wedge I_{q_{i+1}}[u_{a_{i+1}}]$ and $\omega_{i+1}=\omega_i'[u_{a_i}]$
 for $i=0,\cdots, \ell-1$.  Hence,  by the definition of $\mathcal{A}_\tau[\gamma]$, $\xi$ is also a run of $\tau$ under $\gamma$, and thus $R(\mathcal{A}_\tau[\gamma])\neq \emptyset$. 
 
  For the ``if'' direction,  assume there is $\xi$ as defined above which is  a  run of $\tau$ for the  parameter valuation  $\gamma$.
 Then by the  definition of the concrete semantics,  we have $(\gamma, x=0)\models I_{q_0}$, $(\gamma, \omega_i')\models g_{a_{i+1}} \wedge I_{q_{i}} $
  and $(\gamma, \omega_i'[u_{a_{i+1}}])\models I_{q_{i+1}} $ for $i=0.\cdots, \ell-1$.	
  	 In other words,  $(\gamma, \omega_i')\models I_{q_{i+1}}[u_{a_{i+1}}] $ for $i=0.\cdots, \ell-1$.
  	 	 Therefore,  $(\gamma, (0,\cdots,0))\models I_{q_0}$ and
 $\xi$ is a run of $\beta(\tau)$ under $\gamma$, i.e., $R(\mathcal{A}_{\beta(\tau)}[\gamma])\neq \emptyset$. \QEDB
\end{proof}

\subsection{Generating semantic runs from syntactic runs}
We now define the notions of lower bound and upper bound of guards of transitions, and use them to construct feasible runs from  syntactic runs. 

For a PTA $\TA$, we use  $\textit{maxC}(\TA)$   to denote the maximum of the absolute values of the constant terms occurring in the linear expressions of $\mathcal{A}$, that is, 
 \[\textit{maxC}(\mathcal{A})\deff \textit{max}\{|\textit{con}(e)| \mid  e\in \textit{expr}(\mathcal{A})\}. \]  
For a property $\psi$, we use $\textit{maxV}(\psi)$  to denote the maximum of absolute value of the constants which occur in $\psi$, and we define the constant $C\deff 2\cdot \max\{\textit{maxC}({\TA}),\textit{maxV}(\psi)\}+2$.

\begin{lemma}
	\label{lem:one}
	For a one-one  PTA  $\mathcal{A}[x,p]$, a constant $T\ge C$ and a syntactic trace   $\tau=a_1\cdots a_\ell$ of $\TA[x,p]$ such that $R(\mathcal{A}_\tau [p=T])\neq \emptyset$, assume that $-x \prec  -e_1$ (or equivalently $ x\succ e_1$) and  $x\prec e_2$ are two conjuncts  of the guard  $a_i$ for some $i\in \{1, \ldots, \ell\}$. Then $\CF(e_1,p)\le \CF(e_2,p)$. 
\end{lemma}
\begin{proof}  By contradiction. 
\begin{enumerate} 
\item For the parameter valuation  $\gamma=\{p=T\}$,  assume that the lemma does not hold, i.e.  $\CF(e_1,p)> \CF(e_2,p)$. 
\label{item1}
\item  We have $\CF(e_1-e_2,p)>0$. 
\item By the definition of $C$,  $C> |\textit{con}(e_1)|+|\textit{con}(e_1)|$. Then because $T\geq C$, we have  $T>  |\textit{con}(e_1)|+|\textit{con}(e_1)|$. 
\label{item2}
\item   Results \ref{item1})\&\ref{item2}) imply $(e_1-e_2)[\gamma]> 0$.
\label{item3}
\item However,  because  $R(\mathcal{A}_\tau[\gamma])\neq \emptyset$, we have a concrete untimed  run of $\mathcal{A}_\tau [p=T]$
\[(q_0{,}\omega_0){\xrightarrow{a_1}} (q_1{,} \omega_1){\cdots} (q_i{,}\omega_i){\xrightarrow{a_i}} (q_{i+1}, \omega_{i+1}){\cdots} {\xrightarrow{a_\ell} }(q_\ell{,} \omega_\ell)\]
 The guard of $a_i$ holds for $(p=T,\omega_{i})$. Thus, the two conjuncts of the guard of $a_i$ imply that $(e_1-e_2)[\gamma]\le 0$. This contradicts with the result~\ref{item3}).
\label{item4}
\end{enumerate}
	\QEDB
\end{proof}

\begin{corollary}
	\label{coro:up1}
	For a one-one  PTA  $\mathcal{A}[x,p]$, let  $\tau=a_1\cdots a_\ell$  be a syntactic trace of $\TA[x,p]$ such that there is $T\geq C$ for which  $R(\mathcal{A}_\tau [p=T])\neq \emptyset$.
	Assume that for some $i$ and $j$ such that $0\le i<j\le \ell $,  $-x\prec -e_1$  is a   conjunct of the guard of  $a_i$ and  $x\prec e_2$ is a conjunct of  the guard of $a_j$.
	Then,  $\CF(e_1,p)\le \CF(e_2,p)$ if the transitions $a_k$ for $k\in [i,j)$ do not reset any clock, i.e. their reset sets are empty.
\end{corollary}

\begin{definition}[Order between  lower and upper bound constraints]
	\label{def:order}
	For two lower bounds $x\succ_1 e_1$ and $x\succ_2 e_2$ such that  $\succ_1,\succ_2\in \{>,\ge \}$ and a parameter $p$,
	we define that  the order $(x\succ_1 e_1) \sqsupseteq_p (x\succ_2  e_2)$ holds if one of the following  conditions holds. 
	\begin{enumerate}
		\item $\CF(e_1,p)> \CF(e_2,p)$;
		\item $\CF(e_1,p)= \CF(e_2,p) 	\wedge (\textit{con}(e_1)> \textit{con}(e_2))$;
		\item   $e_1 = e_2$ and  $x \succ_1 e_1 \equiv x >e_1$;
		\item $e_1 = e_2$ and $x \succ_2 e_2 \equiv x \ge e_2$.
	\end{enumerate}

	Symmetrically,  let $x\prec_1 e_1$ and $x\prec_2 e_2$ be two upper bounds such that $\prec_1,\prec_2\in \{<,\leq\}$.
	We define that $(x\prec_1 e_1)\sqsubseteq_p (x\prec_2  e_2)$ holds 
	 if one of the following conditions holds. 
	 \begin{enumerate}
	 	\item $\CF(e_1,p)< \CF(e_2,p)$;
	 \item $\CF(e_1,p)= \CF(e_2,p) 	\wedge (\textit{con}(e_1)< \textit{con}(e_2))$;
	 	\item    $e_1= e_2$ and   $(x \prec_1 e_1)\equiv (x < e_1)$;	 	
	 	\item  $e_1 = e_2$ and   $(x\prec_2 e_2) \equiv  (x\le e_2)$.
	 \end{enumerate}
\end{definition}

For a one-one PTA  $\TA[x,p]$, we define  the {\em effective lower bound} of a guard $g$, denoted by  $\ELB(g,p)$, as a syntactic term
\[
\ELB(g,p)\equiv \left \{ \begin{array}{llll} 
                      x>-1 &\mbox{no lower bound $x\succ e$ occurs in  $g$}\\
                     x\succ_1  e &\mbox{if }  x\succ_1  e  \mbox{ is one conjunct of }\\
                     & g  \mbox{ and for all $x\succ_2  e_1$ in $g$,}\\
                     &\mbox{$(x\succ_1  e)\sqsupseteq_p (x\succ_2  e_1)$}
                     \end{array}
                     \right .
                     \]
Symmetrically, we define  the {\em effective upper bound} of a  guard $g$
\[
\EUP(g,p)\equiv  \left \{ \begin{array}{llll} 
                      x<\infty &\mbox{no upper bound $x\prec e$  in $g$}\\
                     x\prec_1  e &\mbox{if }  x\prec_1  e  \mbox{ is one conjunct of } g \\
                      &\mbox{and for all $x\prec_2  e_1$ in $g$,}\\
                     &\mbox{$(x\prec_1  e)\sqsubseteq_p (x\prec_2  e_1)$}
                     \end{array}
                     \right .
                     \]

\begin{lemma}
	\label{lem:noUpdate}
Let $\TA[x,q]$   be a one-one PTA and  
	\[
	\tau=q_0\xrightarrow{g_1\&a_1}q_1\cdots \xrightarrow{g_\ell\& a_\ell} q_\ell\] 
a syntactic run which has no invariants for the locations and no reret sets for the actions.
  
  If there is a $T\ge C$ such that   $R(\mathcal{A}_\tau[p=T])\neq \emptyset$,   $\ELB(g_i,p)\wedge \EUP(g_j,p)$ holds for each valuation $p=t$ such that $t\in[C,\infty)$, where  $g_i$ and $g_j$  are the guards of  $a_i$   and $a_j$  for $i,j \in \{1,\ldots,\ell\}$ such that $i\le j$, respectively.
\end{lemma}

\begin{proof}
Assume $R(\mathcal{A}_\tau[p=T])\neq \emptyset$. Then there is a concrete run of  $\mathcal{A}_\tau[p=T]$ 
	\[\xi = (q_0,\omega_0)\xrightarrow{a_1}(q_1,\omega_1)\cdots \xrightarrow{a_\ell}(q_\ell, \omega_\ell)\]
	Since for all $i,j\in \{1,\ldots, \ell\}$, the transition by $a_i$ does not reset the clock $x$, $\omega_i(x)\leq \omega_j(x)$ if $i\leq j$. Let us set  $\EUP(g_i,p)\equiv x\prec_i e_i$  and $\ELB(g_j,p)\equiv x\succ_j e_j$.
	Since  $\omega_i(x)\le \omega_j(x)$, we have 
	\begin{itemize}
		\item if $\prec_i$ is $<$,   $e_i[p=T]<e_j[p=T]$, else
		\item  if $\succ_j$ is $>$,  $e_i[p=T]<e_j[p=T]$, else
		\item  $e_i[p=T]\le e_j[p=T]$.
	\end{itemize}
	We now make the following two claims:
\begin{enumerate}
\item  if $e_i[p=T]<e_j[p=T]$, $e_i[p=T']<e_j[p=T']$ for $T'\in [C,\infty)$;  and
\item  if $e_i[p=T]\le e_j[p=T]$,  $e_i[p=T'']\le e_j[p=T'']$ for $T''\in [C,\infty)$. We prove these two claims below.
\end{enumerate}
We prove these two claims as follows.
	\begin{enumerate}
		\item In the case when   $e_i[p=T]<e_j [p=T]$, we have $\CF(\ELB(g_i,p),p)\le  \CF(\EUP(g_j,p),p )$ according to {\em Corollary \ref{coro:up1}}.
Hence, $\CF(e_j-e_i,p)\ge0$.  In case when $\CF(e_i-e_j,p)=0$,  $e_i[p=T']<e_j[p=T']$  for $T'\in [C,\infty)$.
		If $\CF(e_j-e_i,p)>0$,  $(e_j-e_i)[p=T']\ge T'-|con(e_i)|- |con(e_i)|>0$ for $T'\in [C,\infty)$. Hence, Claim 1) holds.
		\item The proof for Claim 2) in the  case  when $e_i[p=T]\le e_j[p=T]$, is the same.
\end{enumerate}
	Based on these claims, we prove that for  $T'\in [C,\infty)$.
	\begin{itemize}
		\item if $\prec_i$ is the  relation $<$,   $e_i[p=T']>e_j[p=T']$, else
		\item  if $\succ_j$ is the relation $>$,   $e_i[p=T']>e_j[p=T']$, else
		\item   $e_i[p=T']\ge e_j[p=T']$.
	  	\end{itemize} 
		 Therefore, formula  $\ELB(g_i,p)\wedge \EUP(g_j,p)$ is feasible for $T'\in [C,\infty)$.
	\QEDB
\end{proof}

\begin{definition}[From syntactic to feasible timed run]
For a guard, we define
\[\pr(g,T)\deff \left\{
	\begin{array}{@{}ll@{}} max\{0,e[p=T]\}, \mbox{ if } \ELB(g,p)\equiv x\geq  e   \\
	max\{0,e[p=T]+1\},  \mbox{ if }  \ELB(g,p)\equiv x >e
	\end{array}\right.
	\]
Given a simple  syntactic run which have no location invariants and clock resets 
$\tau=q_0\xrightarrow{g_1\&a_1}q_1\cdots \xrightarrow{g_\ell\&a_\ell }q_\ell$, let
	
	\[
	\PR(\tau, T)\deff (q_0,w_0)\xrightarrow{a_1}(q_1,\omega_1)\cdots \xrightarrow{a_\ell} (q_\ell, \omega_\ell)\]
such that 
	\[
	\omega_i=
	\left\{
	\begin{array}{@{}ll@{}}
	0, & \mbox{if }i=0,\\
	\pr(g,T),  &\mbox{otherwise}.
	\end{array}\right.
	\]
\end{definition}
We now provide Algorithm \ref{the:cvu} for the generation of a feasible timed run from a simple syntactic run. 
\begin{algorithm}[!htb]
	\DontPrintSemicolon
	\SetKwData{Left}{left}\SetKwData{This}{this}\SetKwData{Up}{up}
	\SetKwFunction{Union}{Union}\SetKwFunction{FindCompress}{FindCompress}
	\SetKwInOut{Input}{input}\SetKwInOut{Output}{output}

	\Input{ A simple syntactic run of  $\TA[x,p]$  $\tau=q_0a_1q_1\cdots a_\ell q_\ell$ such that there is a $T_1\ge C$ and $R(\tau[p=T_1])\neq \emptyset$; an integer  $T\ge C$.}
	\Output{ A run $\xi=s_0d_0s_0'a_1s_1\cdots d_{\ell-1}s_{\ell-1}'a_\ell s_\ell $ is a run of $\TA[p=T]$. }

	\SetAlgoLined
	\BlankLine
	Set  $s_0a_1s_1\cdots a_\ell s_\ell =\PR(\tau, T)$ where $s_i=(q_i, \omega_i)$\;
	\label{init}
	Set  $\xi_1=s_0a_1s_1\cdots a_\ell s_\ell $\;
	\While{$\xi_2$ is not a simple feasible sequence}
	{	\label{body}
		\For{ $i\in [1,\ell]$  }{	
			\If{ $\omega_i<\omega_{i-1}$}{
				Set $\omega_i=\omega_{i-1}$;
				\label{update}
			}
		}
	}
	\For{ $i\in [1,\ell]$  }
	{
		$d_{i-1}=\omega_i-\omega_{i-1}$;
		$s_{i-1}'=(q_{i-1},\omega_i)$;
		
	} 
	\Return{$s_0d_0s_0'a_1s_1\cdots d_{\ell-1}s_{\ell-1}'a_\ell s_\ell$\;}
	\caption{GSR (Generate Simple Feasible Run)\label{the:cvu}}
\end{algorithm}

\begin{lemma}
	\label{lem:term}
	Algorithm \ref{the:cvu} terminates within a finite number of  steps. When it terminates, $(p=T,\omega_i)\models g_i$ holds for the output   $\xi=s_0d_0s_0'a_1s_1\cdots d_{\ell -1}s_{\ell-1}'a_\ell s_\ell$,  where $g_i$ is the guard of transition $a_i$ for $i=1,\cdots, \ell$.
\end{lemma}  

\begin{proof}

To prove the termination, let $M=\max\{w_i\mid i=0,\cdots, \ell\}$. It is easy to see that $(\omega_0\le M\wedge \cdots,\omega_\ell \le M)$ is an invariant of the while loop, i.e. it holds when the execution enters line \ref{body}.  Each iteration of the loop body increase at least one of $\omega_i$ by the execution of the statement of line \ref{update}. Hence, the algorithm terminates within a finite number of iterations, as otherwise the invariant would be falsified.
		
We prove the correctness of the algorithm by contradiction. Assume that there is an $i\in\{1,\ldots,\ell\}$  such that   $(p=T,\omega_i)\not \models g_i$. By the definition of $\PR(\tau,T)$, $(T,\omega_i)\models g_i$ holds after the execution of the statement in  line \ref{init} of {\em Algorithm \ref{the:cvu}}.

It is noticed that $\omega_i$ is possibly changed only by the statement $\omega_i:=\omega_{i-1}$ in  line \ref{update}, which increases $\omega_i$.  Therefore, when the algorithm  terminates for any $i\in\{1,\ldots, \ell\}$, $\omega_i =  \omega_k^0$  for some $k\leq i$, where $\omega_k^0$ is the initial   value of $\omega_k$. Assume $(T,\omega_i)\models g_i$ does not hold, that is, $(T,\omega_i)\not \models g_i$ holds.   This implies that   $ (T,\omega_i)\not \models \EUP(g_i,p)$. By the definition of $\pr(g_k,T)$, $\pr(g_k,T)=\omega_k^0$ is the minimum value  of $x$ which make formula $(T, x)\models g_k$ hold.
	According to  {\em Lemma \ref{lem:noUpdate}}, formula  $\ELB(g_k,p)\wedge \EUP(g_i,p) $ is feasible for $p\in[C,\infty)$.  Because $\EUP(g_i,p)$ is an upper bound constraint and $\omega_i=\omega_k^0$, $(T, \omega_i)\models \ELB(g_k,p)\wedge \EUP(g_i,p)$. Which
	contradicts with the assumption  that $ (T, \omega_i)\models  \EUP(g_i,p) $ does not hold. Therefore  $(T, \omega_i)\models  \EUP(g_i,p) $ must hold. This implies  $(T,\omega_i)\models g_i$   holds for $i=1,\cdots, \ell$. \QEDB
\end{proof}

\begin{lemma}
	\label{the:oneS}
	Let $\tau$ be a simple syntactic  run of a one-one PTA $\mathcal{A}[x,p]$.   We have $R(\mathcal{A}_\tau[p=T])\neq \emptyset$ for any $T\ge C$ if there is a $T_1\ge C$ such that $R(\mathcal{A}_\tau[p=T_1])\neq \emptyset$.
\end{lemma}
\begin{proof}
{\em Algorithm \ref{the:cvu}}  generates a run of $\mathcal{A}_\tau[p=T]$ for  simple syntactic  run $\tau$. \QEDB
\end{proof}

\begin{lemma}
	\label{lem:new1}
	Let $\TA[x,p]$ be a one-one  PTA  and $\psi$  a formula of the form  $\exists \Diamond \phi$.   Then,    $\TA[p=T]\models  \psi$ for all  $T\ge C$  if there is a $T_1\ge C$ such that  $\TA[p=T_1]\models \psi$,.
\end{lemma}

\begin{proof}
	For the parameter valuation  $\gamma=\{p=T_1\}$, let  $\phi=q_\ell$ without the loss  of generality, as the proof will be similar when $\phi$ is in other forms. Assume $\TA[p=T_1]\models \psi$, we  need to prove that  $\TA[p=T]\models  \psi$ for any  $T\ge C$.
	
Let $\tau=q_0\xrightarrow{a_1}q_1\cdots \xrightarrow{a_\ell} q_\ell$  be a syntactic   run of $\mathcal{A}[p=T_1]$  that satisfies property $\psi$, i.e. $R(\mathcal{A}_\tau[p=T])\neq \emptyset$. According to {\em Lemma \ref{lem:moveI}},
	we  obtain an untimed  run $\beta(\tau)$ without location invariant  which satisfies that  $R(\mathcal{A}_{\beta(\tau)}[p=T_1])\neq \emptyset$ if and only if   $R(\mathcal{A}_\tau[p=T_1])\neq \emptyset$.   Without the loss  of generality, set
	\[ \beta(\tau) =q_0\xrightarrow{g_1\&a_1[u_1]}q_1\cdots \xrightarrow{g_{\ell}\&a_{\ell}[u_{\ell}]}q_{\ell}.\]
	
	Let $k$ be the number of transitions in $\beta(\tau)$ which reset the clock $x$. We prove the lemma by induction on $k$.

For $k=0$,  $\beta(\tau)$ is a simple untimed run and $R(\TA_{\beta(\tau)}[p=T])\neq \emptyset$ follows from {\em Lemma \ref{the:oneS}}.

We assume $R(\TA_\tau[p=T])\neq \emptyset$ holds for $k_0$ and  let $k=k_0+1$. Assume $a_i$ is the action for the last transition in $\beta(\tau)$ that resets the clock $x$, and let 
\[\tau_2=q_0\xrightarrow{g_1\&a_1[u_1]}q_1\cdots \xrightarrow{g_{i}\&a_{i}[u_{i}]}q_{i}\]
 We obtain the sequence $\tau_2'$
\[\tau_2'=q_0\xrightarrow{g_1\&a_1[u_1]}q_1\cdots \xrightarrow{g_i\& a_i}q_i\]
by removing the last reset set $u_i$ of $\tau_2$. 

By the induction hypothesis and $\tau_2'$ has $k_0$ transitions that modify  $x$, $R(\mathcal{A}_{\tau_2'}[p=T])\neq \emptyset$  for any $T\ge C$.   Since $\tau_2'$ is the same as $\tau_2'$ excpet the reset set of last transition,  $R(\mathcal{A}_{\tau_2}[p=T])\neq \emptyset$ for any $T\ge C$.
The value of clock $x$  is a fixed value when the run reaches $q_i$ in $\tau_2$, sine  there is a reset of
$x$ in transition $a_i$. And this value is the same as the value of clock $x$
when  the run reaches $q_i$ in the run $\beta(\tau)$ under the parameter valuation $p=T_1$.  As there is 
no reset of $x$ in transitions $a_{i+1},\cdots, a_\ell$. Hence,
$R(\mathcal{A}_{\beta({\tau})}[p=T])\neq \emptyset$ for $T\ge C$.

	\QEDB
\end{proof}
\subsection{The proof of the main theorem}
We can now prove  {\bf Theorem}~\ref{th:one-one-PTA} of this section. 
\begin{proof}
	Assume that  $\TA[p=C]\models \psi$.  Following  Lemma \ref{lem:new1}, we initially start with the subset of parameter valuations $H=\{C,C+1,\cdots\}\subseteq \Gamma(\TA,\psi)$. We then  
	iteratively check if $\TA[p=i]\models \psi$  holds for $i=0.\cdots,C-1$ and add  to $H$ those $i$'s such that $\TA[p=i]\models \psi$ holds. This procedure terminates with $H = \Gamma(\TA,\psi)$. \QEDB
\end{proof}

\begin{corollary}
	\label{coro:new5}
	Let $\TA[x,p]$ be a one-one PTA . The set $\Gamma(\TA,\psi)$ is solvable if $\psi$  is the form  $ \forall \Box \phi$.
\end{corollary}

\section{Parameter  synthesis problem for L/U-automata}
\label{sec:newresult}
%

In this section, we will consider the parameter synthesis problem of  {\em L/U-automata} which defined in \cite{HUNE2002183} as given below.
\begin{definition}[L/U automata]
	Let $e=c_0+c_1p_1+\cdots+c_np_n$  be a linear expression. For $i=1,\cdots, n$, we say $p_i$ occurs  in $e$ if $c_i\neq 0$,  occurs positive in $e$ if $c_i>0$, and occurs negative in $e$ if $c_i<0$. 
	\begin{itemize}
		\item 	A parameter $p$ of PTA $\TA$ is a lower bound (or an upper-bound) parameter if it only occurs negative (resp. positive) in the expressions of $\TA$. 
		\item $\TA$ is called a lower-bound/upper-bound (L/U) automaton if  every parameter of $\TA$ is either a lower-bound parameter or an upper-bound parameter.	
	\end{itemize}
\end{definition}

For instance, $p_1$ is an upper bound parameter in $x-y< 2p_1$;  $p_2$ and 
$p_3$ are lower bound parameters in $y-x< -p_2- 3p_3$ and in $x-y< 2p_1-p_2-2p_3$. A PTA which contains both the constraints $x-y\le p_1-p_2$ and $z< p_2-p_1$ is not an L/U automaton.

\subsection{Parameter  synthesis for  L/U-automata}
Clearly,  the  parameters in a PTA $\mathcal{A}$ can be divided  into two $L(\mathcal{A})$ and $U(\mathcal{A})$ which are the  sets lower-bound parameters and upper-bound parameters, respectively\footnote{We simply use $L$ and $U$ when there is no confusion.}. For a parameter valuation $\gamma$ 
we use  $\gamma_l$ and $\gamma_u$ to denote its restrictions on $L$ and $U$, respectively.   The following proposition  in  \cite{HUNE2002183} is useful for us. 

\begin{proposition}
	\label{prop:monot}
	Let $\TA$ be an L/U automaton and $\phi$ a state formula. Then
	\begin{enumerate}
		\item  $\TA[\gamma_l{,}\gamma_u]{\models} \exists \Diamond \phi$ iff  $\forall \gamma_l'{<} \gamma_l{, }\gamma_u{<} \gamma_u'{:}\TA[\gamma_l'{,}\gamma_u']{\models}\exists \Diamond \phi$.
		\item $\TA[\gamma_l{,}\gamma_u]{\models} \forall \Box \phi$ iff $\forall \gamma_l{< }\gamma_l'{,} \gamma_u'{<} \gamma_u{:}\TA[\gamma_l'{,}\gamma_u']{\models}\forall \Box \phi$.
	\end{enumerate}
\end{proposition}

The proof of this proposition given in  \cite{HUNE2002183}  needs to extend  the notion of a parameter valuation to that of a  {\em partial parameter valuation}  which allow a parameter to be ``undefined''. We use $\infty$ to denote the undefined value. Thus, a partial valuation $\gamma$ assigns a   parameter with a  value in 
$\NN\cup \{\infty\}$, rather than in $\NN$ only. 

Partial  parameter valuations are useful in certain  cases to solve the verification problem. However partial parameter valuations may cause problems. For example,  if $\gamma[p_1]=\gamma[p_2]=\infty$, what would be the value of $\gamma(e_1-e_2)$? To avoid this problem, we require that a partial parameter  valuation does not assign  $\infty$ to both a lower-bound parameter and an upper-bound  parameter. Also we follow the conventions that the truth values of  $0\cdot \infty=0$, and $x- y\prec \infty$ 
are {\em true} and  the truth value of $x-y\prec -\infty$ is {\em false}. We use  $[0,\infty]$ to denote the valuation which assigns $0$ for each lower bound parameter and $\infty$ to each upper bound parameter. 

We now show that  the {\em emptiness problem} of  an L/U automaton can be reduced to  the reachability problem of  its corresponding timed automaton under parameter valuation $[0, \infty]$.

\begin{proposition}
	\label{prop:new1}
	Let $\TA$ be an L/U automaton and $\phi$ be a  state formula. Then $\TA[0,\infty]\models \exists \Diamond \phi$ if and only if 
	there exists a  parameter valuation and clock evaluation $\omega$  such that  $\TA[\gamma,\omega]\models \exists \Diamond \phi$.
\end{proposition}
\begin{proof}
	The ``if" part is an immediate consequence of {\em Proposition  \ref{prop:monot}}. For the ``only if" part, assume that $\xi$ is a run of $\TA[0,\infty]$ which  satisfies $\phi$.  Let $T$ the maximum clock value occurring in $\xi$ and $T'$ be the smallest constant occurring in $\TA$ and $\phi$. More precisely, if $\xi=(q_0{,}\omega_0){\xrightarrow{a_1}}(q_1{,}\omega_1){\xrightarrow{a_2}}{\cdots} {\xrightarrow{a_\ell}}(q_\ell{,}\omega_\ell)$, then
	$T=\max\{\omega_i(x)~|~ 0\le i\le \ell, x\in X\}$.     Let  $\gamma_l(p)=0$ for $p\in L$ and $\gamma_u(p)=T+|T'|+1$ for $p\in U$. The proposition can then be proven by considering the different possible cases of the location invariants and guards of the transitions in $\xi$. For example, assume    $g=x-y\prec e$  is the  invariant of a location  location $q_i$, or a conjunct of the guard of transition by $a_i$, or a conjunct of $\phi$.  The relation $\omega_i(x)-\omega_i(y)\prec e[\gamma_l,\gamma_u ]$ holds for the definition of $\gamma_l$ and $\gamma_u$.  Hence, $((\gamma_l, \gamma_u),\omega_i)\models g$.  Thus, $\xi$ is a run of $\TA[\gamma_l, \gamma_u]$  and   $\TA[\gamma_l, \gamma_u]\models \exists\Diamond \phi$.  \QEDB
\end{proof}
Proposition \ref{prop:new1} provides an algorithm to check the satisfaction of a property with existential quantifiers  by an L/U-automaton. Base on  the ``monotonic" property of L/U-automata,  this actually  reduces the {\em emptiness problem}   an L/U-automaton  to  the reachable problem of corresponding timed automaton.
\begin{lemma}
	\label{lem:up}
	For  a one-one  L/U PTA  $\TA[x,p]$ and a   formula   $\psi\equiv  \exists \Diamond \phi$,  if there exists $T\ge C$ such that
	 $\TA[p=T]\models \psi$,  then set  $\Gamma(\TA, \psi)$ is computable. 
\end{lemma}
\begin{proof}
	Suppose  there is a syntactic  run \[\tau=(q_0,I_{q_0})\xrightarrow{g_{1}\&a_1[u_{1}]}(q_1,I_{q_1})\cdots \xrightarrow{g_\ell\&a_\ell[u_\ell]}(q_\ell,I_{q_\ell})\] which  satisfies $\psi$ under the parameter valuation $p=T$.  According {\em Lemma \ref{lem:new1}},  $\TA[p=T_1]\models \psi$, for 
	$T_1\ge C$. Then, we can check wether $\TA[p=T_1]\models \psi$ for $T_1\in [0,C]$. Therefore,  set  $\Gamma(\TA, \psi)$ is computable. 
	
\QEDB
	 		
\end{proof}

\begin{proposition}
	\label{prop:oneP}
	For  a one-one L/U automaton $\TA[x,p]$ and state property $\psi$, the set $\Gamma(\TA, \psi)$ is computable for $\psi\equiv \exists \Diamond \phi$.
\end{proposition}
\begin{proof}
	Let $H$ be $\Gamma(\TA, \psi)$. Assume that  $p$ is a  lower bound parameter. First check whether  $\TA[p=0]\models \psi$ hold or not. Since $\TA[p=0]$ is a \taa, 
	this checking is decidable.
	If $\TA[p=0]\models \psi$ does not hold, then employ {\em Proposition \ref{prop:new1}}, $H=\emptyset$. Otherwise,  
	  if $\TA[p=C ]\models \psi$ does not hold, then employ {\em Lemma \ref{lem:up}},
	$H\subseteq\{0,1,\cdots, C-1\}$. We can check whether  $\TA[p=i]\models \psi$ holds or not from  $i=C-1$ to $i=0$ until formula holds,
	then $H=\{0,\cdots, i\}$.  	If $\TA[p=C]\models \psi$ holds, $H$  is solvable follows  from {\em Lemma \ref{lem:up}}.
	
	If $p$ is an upper parameter.  First check whether $\TA[p=\infty]\models \psi$ holds or not. Since $\TA[p=\infty]$ is a \taa, this checking is decidable.
	If $\TA[p=\infty]\models \psi$ does not hold,  $H=\emptyset$ follows from {\em Proposition \ref{prop:new1}}.
	Otherwise,	 assuming that $\xi$ is a run of 
	$\TA[0,\infty]$ that  satisfies $\phi$. 
	 Let $T'$ be the smallest constant occurring in $\TA$ and $\phi$. And let $T$
	 be the maximum clock value occurring in $\xi$. More precisely, if $\xi=s_0\xrightarrow{a_1}s_1\cdots \xrightarrow{a_\ell} s_\ell$ and $s_i=(q_i,\omega_i)$, then
	 $T=\max_{i\le \ell, x\in X}\omega_i(x)$
 It is easy to check that 
	$\{T+|T'|+1,\infty\}\subseteq H$.   We   iteratively check whether  $\TA[p=i]\models \psi$ holds or not from  $i=0$ to $i=T+|T'|+1$ until formula holds,
	then $H=\{i,i+1,\cdots, \infty\}$.  \QEDB
	
\end{proof}

For an  L/U automaton $\TA$ with one parameter $p$ and a property $\psi$,  the 
work in \cite{bozzelli2009decision} shows that  the complexity of  computing $\Gamma(\TA,\psi)$
is PSPACE-complete.  

\begin{corollary}
	\label{coro:new3}
	For  a one-one L/U PTA  $\TA[x,p]$  and state property $\psi$, the set $\Gamma(\TA, \psi)$ is computable for $\psi=\forall \Box \phi$.
\end{corollary}
\begin{proof}
	Let $H$ be $\Gamma(\TA,\psi)$ and  $H_1$ the set of parameter valuations which make 
	$\TA[\gamma]\models \exists \Diamond \neg \phi$ hold. It is easy to know that $\NN=H\cup H_1$ and $H\cap H_1=\emptyset$. By {\em Proposition \ref{prop:oneP}}, $H_1$	is computable, hence $H=\NN\setminus H_1$ is also computable.  \QEDB
\end{proof}

\begin{theorem}
	\label{prop:all}
	For 
 a L/U PTA  $\TA$ with  one parametrically
constrained clock, the  set $\Gamma(\TA,\psi)$    is computable if  $\psi\equiv \exists \Diamond \phi$ and all the parameters are lower bound parameter or  all the parameters are upper bound parameter.
\end{theorem}
\begin{proof}
	Let $H=\Gamma(\TA,\psi)$. 
	If  all the parameters are lower bound parameter, we construct a PTA $\TA'$  from $\TA$ by  replacing all the parameters of  $\TA$ by the single  parameter $p$. 
	Then $\TA'$ is be an L/U automaton with one  parameter $p$. Employing {\em Proposition \ref{prop:oneP}}, we can compute a  set $H'=\Gamma(\TA',\psi)$. \ly{When $H'=\emptyset$,  by {\em Proposition \ref{prop:monot}}, $(0,\cdots,0)\in H$ if $H\neq \emptyset$. Hence $H=\emptyset$.} When $H'=\NN$, i.e.,   after setting each parameter of $\TA$ to $\infty$, $\psi$ also holds. Hence, $H=\NN^m$. Otherwise, by the 
        {\em Proposition \ref{prop:oneP}}, there is a $T\ge 0$ such that $H'=\{0,1\cdots, T\}$. 
	Assuming that there is   parameter valuation $\gamma$ such that $\gamma(p_i)\ge T+1$ for all $i=1,\cdots, m$ such that $\TA[\gamma]\models 
	\psi$,
	by the {\em Proposition \ref{prop:monot}}, $(T+1,\cdots, T+1)\in H$ and $T+1\in H'$,	which contracts with $H'=\{0,\cdots, T\}$. Hence, the assumption does not hold and there exits at
	least one component of $\gamma$ which less or equal than  $T$ for each $\gamma \in H$. Let $\TA'_{ij}$ be a PTA which obtain from set parameter $p_i$
	to  $j$.  Let $H'_{ij}$ be set of parameter valuation 
	such that  $ \TA_{ij}'[\gamma]\models \psi$. We lift $H_{i,j}'$ to $H_{i,j}$ by
	let the $i$-th of $H_{i,j}$ be $j$ and other components be the same as $H_{i,j}'$. In other words,
		\[H_{ij}=\{\gamma\mid \left(\gamma[p]=\gamma_1[p] , p\neq p_i\right)\wedge \left(\gamma(p_i)=j\right),  \gamma_1\in H_{ij}'   \}.\]
	 Then $H=\bigcup_{i=1}^m\bigcup_{j=0}^{T}H_{ij}$.
	
	If  all the parameters are upper bound parameter, we construct  a PTA $\TA'$  from $\TA$ by replacing all the parameters of $\TA$ by the single parameter $p$.  Then $\TA'$ is
	an L/U automaton with one parameter. Let $H'$ be $\Gamma(\TA',\psi)$. Following  {\em Proposition \ref{prop:oneP}}, we can compute
	$H'$. When $H'=\emptyset$.  From {\em Proposition \ref{prop:monot}}, $(\sum_{i=1}^m \gamma(p_i),\cdots,\sum_{i=1}^m \gamma(p_i))\in H'$ for each $\gamma\in H$. Hence $H=\emptyset$. \ly{When $H'=\NN$, in the other words,  after setting each parameter of $\TA$ to $0$, $\TA[p=0]\models \psi$ also holds.  Hence, $H=\NN^m$.} If $H'\neq \empty $ and $H'\neq \NN$, by the {\em Proposition \ref{prop:oneP}}, there is a $T>0$ such that $H'=\{T,T+1,\cdots,\infty\}$.  By the definition of $H'$, $(T,\cdots, T)\in H$, therefore, $\{\gamma\mid \gamma(p_i)\ge T, i=1,\cdots,m\}\subseteq H$. Let $\TA'_{ij}$ be a PTA which obtain from set parameter $p_i$
	by $j$. 
	Let $H'_{ij}$ be $\Gamma(\TA'_{ij},\psi)$. We lift $H_{i,j}'$ to $H_{i,j}$ by
		\[H_{ij}=\{\gamma\mid \left(\gamma[p]=\gamma_1[p] , p\neq p_i\right)\wedge \left(\gamma(p_i)=j\right),  \gamma_1\in H_{ij}'   \}.\]
	  Then $H=\left(\bigcup_{i=1}^m\bigcup_{j=0}^{T-1}H_{ij}\right)\cup \{\gamma\mid \gamma(p_i)\ge T, i=1,\cdots,m\}$, since there is at least one component of $\gamma$ less than $T$ excpet $\{\gamma\mid \gamma(p_i)\ge T, i=1,\cdots,m\}$. \QEDB

\end{proof}

\begin{remark}
When all the parameters are lower-bound parameter or all the parameters are upper-bound parameter,
the authors in \cite{bozzelli2009decision}  provide a method to  compute the explicit representation of the set of parameter valuation for
which there is a corresponding infinite accepting run of the automaton. Our  result
is concerning on more general   
properties.
\end{remark}
The following corollary directly follows Theorem~\ref{prop:all}. 
\begin{corollary}
	\label{coro:new4}
	For an  L/U automaton $\TA$ with  one parametrically
constrained clock and a propery  
	$\psi\equiv \forall \Box \phi$, the  set $\Gamma(\TA,\psi)$ is computable if all the parameters are lower-bound parameters or  all the parameters are upper-bound parameters.
\end{corollary}

\section{A learning algorithm for L/U automata }
\label{sec:furth}

Corollary~\ref{coro:new4} only applies to   either L-PTAs or  U-PTAs. We intend to  tackle a more general class of L/U PTAs, for which the parameter synthesis problem is known unsolvable \cite{jovanovic2015integer}. In this section, we, instead sacrifice the completeness of the algorithm to computer the exact set $\Gamma(\TA,\psi)$ of the feasible parameter valuations,  propose a learning based approach to identify the boundary of the region $\Gamma(\TA,\psi)$. 

\subsection{Connectedness of the  region  $\Gamma(\TA,\psi)$}

In what follows, we will show a topological property for $\Gamma(\TA,\psi)$ where $\TA$ is an L/U automaton and the property has the form $\exists \Diamond \phi$. In general for a point in the $m$ dimensional space $\NN^m$, use $v(i)$ to denote $i$th dimension  $v$ and $|v-v'|$ to denote the distance between $v$ and $v'$.

\begin{proposition}
	\label{prop:connect}
	For two  $\gamma_1$ and $\gamma_2$ two feasible parameter valuations of  $\TA$  and  $\psi = \exists \Diamond \phi$,  there exists a sequence of lattice points $v_0,\cdots,v_k$ which are feasible parameters for $\TA$ and $\psi$ such that $\gamma_1=v_0,\gamma_2=v_k, |v_i-v_{i-1}|=1$. \end{proposition}
We say  the sequence $v_0,\cdots,v_k$ in the proposition  {\em connects} $\gamma_1$ and $\gamma_2$.

\begin{proof}
Let us use $H$ to denote  the feasible region $\Gamma(\TA, \psi)$, and we make the proof  by induction on  the number $m$ of parameters. 
\begin{itemize}
\item For $m=1$, we use  {\em Proposition \ref{prop:oneP}}. If  $H$ is not empty, it can be one of the three sets $\NN$, $\{0,1,\cdots,T\}$ or $\{T,T+1,\cdots\}$. It is clearly that for any two feasible valuations in either these three sets we can find a sequence of the points satisfying the conditions in the proposition.
\item Assuming that the proposition holds for all $m\le M$.  
\item We need to prove the proposition holds for $m=M+1$. We divide the proof into two cases when $\TA$ has  no lower bound parameter, and when it has lower bound parameters. 
\begin{itemize}
\item If $\TA$ has  no  lower bound parameters, let $T=\max\{ \sum_{i=1}^m \gamma_1(i),\sum_{i=1}^m \gamma_2(i) \}$. According to  Proposition \ref{prop:monot},  we know $(T,\cdots, T)\in H$. We  repeatedly  add $1$ to  $\gamma_1$'s $i$-th item  until $\gamma_1(i)=T$ and repeatedly  add $1$ to    $\gamma_2(i)$ until $\gamma_2(i)=T$ for $i= 1, \ldots, m$.  This procedure generates the sequence of points we search for feasible parameter valuations. 
\item If $\mathcal{A}$ has lower bound parameters, let $p_i$ be a lower bound parameter of $\mathcal{A}$. We generate a sequence $s_1$ of points by repeatedly decrementing $\gamma_1(i)$ and $\gamma_2(i)$ by $1$, respectively, until $\gamma_1(i)=0$ and $\gamma_2(i) =0$. We use   $\gamma_1^{p_i=0}$ and  $\gamma_2^{p_i=0}$ to denote  two points in the $M$-dimensional parameters space obtained from $\gamma_1$ and $\gamma_2$ by removing the dimension for  $p_i$, respectively. Let $\TA'=\TA[\{p_i=0\}]$. It is easy to know that $\gamma_1^{p_i=0},\gamma_2^{p_i=0}\in \Gamma(\TA',\psi)$. By the induction assumption, the proposition holds for $m=M$ and exits a sequence $s_3$ connect $\gamma_1^{p_i=0}$ and $\gamma_2^{p_i=0}$. Then it is easy to see $s_1s_3s_2$
	is a sequence which connects original $\gamma_1$ and $\gamma_2$.  \QEDB
\end{itemize} 
 \end{itemize}

\end{proof}
Proposition \ref{prop:connect} says that the $\Gamma(\TA, \psi)$ for PTA $\TA$ and property $\psi::=\exists \Diamond \phi$ is a single ``single connected" set.  Informally,  as the $\Gamma(\TA, \psi)$
only consider lattice points, the meaning of ``single connected" set  is that each pair points $(v,v')\in \Gamma(\TA, \psi)$ can connect by  near points in $\Gamma(\TA, \psi)$.

\subsection{A learning based algorithm}
 As we have seen from the previous section,   $\Gamma(\TA,\psi)$ is a   {\em connected  set} when  $\psi$ contains existential quantifiers only. 
 We treat  the problem of identifying the boundary of  $\Gamma(\TA, \psi)$ as 
 the problem of two-class classification in machine learning  where a {\em decision surface} (or {\em decision boundary})  is computed to separate the feasible and infeasible parameter valuations. 
 To this end, we design an algorithm which combines the geometric concepts learning algorithm proposed in \cite{sharma2012interpolants} and binary classifier support-vector machine  (SVM).
 Intuitively, a   parameter valuation is ``good'' if it is feasible for the given PTA and the given properties and a  ``bad'' parameter valuation, otherwise. We describe the algorithm as follows.

 \noindent \textbf{\emph{Step 1}} (\textbf{\emph{Initial Parameter Generation}}): 
    A Monte Carlo method is  used to repeatedly generate a pair sets of good and bad parameter valuations $G_i,B_i$ in $i$-th round where   $G_i=\{g_{i1}, \ldots, g_{in_{i}}\}$ and $B_i=\{b_{i1}, \ldots, b_{in_{i}}\}$ be the set of good and bad points, respectively, and $A_i=G_i\cup B_i$.
For the case of  2-dimension, the good points $G_i$ and bad points $B_i$ are illustrated  in Fig. \ref{fig:svm}.
         
          \begin{figure}[h!]
          	\centering
          	\includegraphics[width=0.33\textwidth]{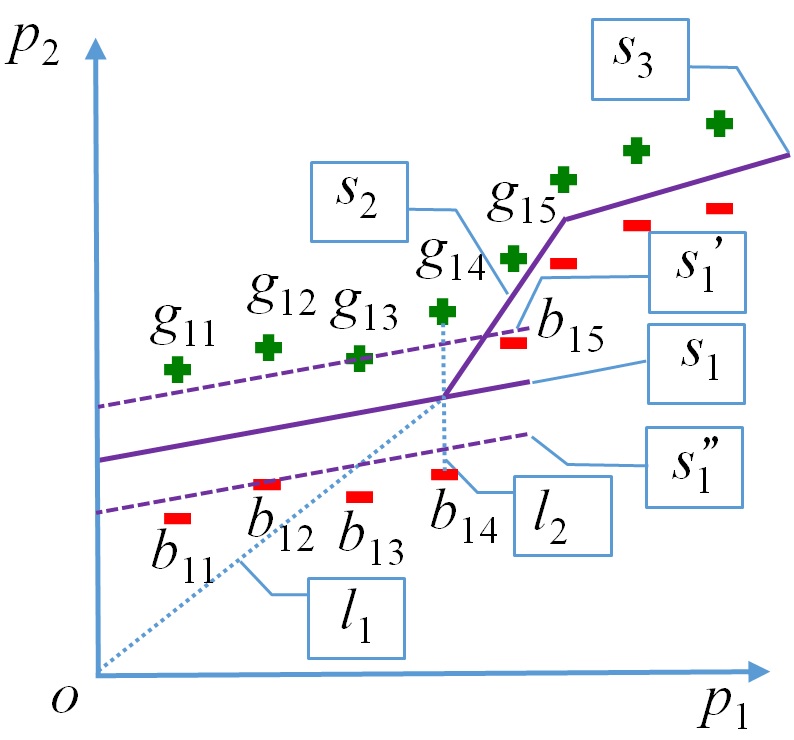}
          	\caption{SVM. \label{fig:svm}}
          \end{figure}

\noindent \textbf{\emph{Step 2}} (\textbf{\emph{Classification by SVMs}}): 

We create a bipartite graph $G$ where good points lie above the  bad ones lie in lower and each edge, such as  $l_2$ in Fig.~\ref{fig:svm},  connects  a pair of good and bad points. 

\smallskip

\noindent \emph{Step 2.1}: In the $i$-th round, for the points in the set  $A_i$, a SVM is used  to maximize the distance between the nearest training data points, and computes the maximum margin classification between the good and bad points.   For the case of 2-dimension,  Fig. \ref{fig:svm} shows the separation between the good and bad points in $A_1$ as follows
\begin{itemize}
\item $s_1$ is the \emph{maximum-margin hyperplane}; 
\item the distance between  lines $s_1'$ and  $s_1''$ is the \emph{maximum-margin}; and 
\item the points, such as  $g_{1,3}$ and $b_{1,2}$, which  lie on line $s_1'$ or $s_1''$ are the  {\em support vectors}. 
\end{itemize}
Thus, we have obtained  the maximum-margin hyper-plane $s_1$ as the separation boundary (i.e. decision surface), i.e. $s_1'$ and $s_1''$ have the maximum-margin.  According to the property of SVM, the maximum-margin ensures that  the separation that has the highest generalization ability.
The  learning process of SVM also has  a strategy that if there does not exists a single hyperplane to  separate all good and bad points in $A_1$, the points with smaller indices  are separated before the points  with larger indices.

\smallskip

\noindent \emph{Step 2.2}: We represent $A_i$ as the array $A_i=((g_{i1},b_{i1}), \ldots,  (g_{i n_i},b_{i n_i}))$. If single hyperplane  cannot be  found in \emph{Step 2.1} to separate $G_i$ and $B_i$ in $A_i$, we check through  the vector $A_i$ from the left to the right to find the first pair, say $(g_{ik}, b_{ik})$   that cannot be separated (e.g. $(g_{1,5}, b_{1,5})$ in Fig.~\ref{fig:svm}). We call the pairs before $(g_{ik}, b_{ik})$ {\em covered} (by the checking process)~\cite{sharma2012interpolants}, and denote it as $C_i$, and the rest pairs of $A_i$ are {\em uncovered} and denoted as $U_i$. For example, $C_1=\{ g_{11}, g_{12}, g_{13}, g_{14}, b_{11}, b_{12}, b_{13}, b_{14} \}$ in  Fig~\ref{fig:svm}. We take the hyperplane (e.g. $s_1$ in Fig.~\ref{fig:svm}) that separates $C_i$ as a {\em boundary segment} for separating the good and bad points of $A_i$.  Then we set the array $A_i$ as the $U_i$ and go back the SVM process in  \emph{Step 2.1}. If a hyperplane is found to separate these pairs, add the hyperplane to the boundary segments which have found, and repeat the process, otherwise.
	
\smallskip
   
\noindent \textbf{\emph{Step 3}} (\textbf{\emph{Continuous Parameters Generation and Classification}}): In the same way as  the boundary is obtained by  previous two steps, we generate a new set of pairs of good and bad points in the  new round.  
Each good or bad point is generated near the  boundary generated by the end of {\em Step 2} with the distance of the point to its nearest hyperplane being a given margin $w$  ($w$ can be assigned to be 1 in consideration of the integer-related feature of our problem). 
The algorithm repeats \emph{Step 2} to generate refined  boundaries until reaching a given number of iterations or meeting some given criteria for  termination.

\section{Conclusion}
    \label{sec:con}
We have studied the parametric synthesis problem for   parametric timed automata.
We  
have provided an algorithm to construct the feasible parameter region when PTA with one parametric clock and 
one parameter.
 We have proved  that, if PTA is restricted to be with only lower-bound
    or upper-bound parameters, the   parametric synthesis problem  is solvable. 
   Furthermore,  we have shown  that the feasible parameter region of more general L/U automata is a ``single connected" set for a property  which contains existential quantifiers only. 
     Aided by this result, we have presented  a SVM based method to compute the boundaries  of  feasible parameter regions.  
     
     In the further,  in the theorem phase, we will extend decidable result of parameter synthesis problem in PTA with one parametrically constrained clock and many paramters. In the algorithm phase,  we will give the experience result of our algorithm in some test cases.

\section*{Acknowledgements}
The authors would like to thank Andr{\'e} {\'E}tienne   who give   us
many meaningful suggestions.

\bibliographystyle{IEEEtran}
\bibliography{IEEEabrv,ptm}

\begin{thebibliography}{10}
\providecommand{\url}[1]{#1}
\csname url@samestyle\endcsname
\providecommand{\newblock}{\relax}
\providecommand{\bibinfo}[2]{#2}
\providecommand{\BIBentrySTDinterwordspacing}{\spaceskip=0pt\relax}
\providecommand{\BIBentryALTinterwordstretchfactor}{4}
\providecommand{\BIBentryALTinterwordspacing}{\spaceskip=\fontdimen2\font plus
\BIBentryALTinterwordstretchfactor\fontdimen3\font minus
  \fontdimen4\font\relax}
\providecommand{\BIBforeignlanguage}[2]{{%
\expandafter\ifx\csname l@#1\endcsname\relax
\typeout{** WARNING: IEEEtran.bst: No hyphenation pattern has been}%
\typeout{** loaded for the language `#1'. Using the pattern for}%
\typeout{** the default language instead.}%
\else
\language=\csname l@#1\endcsname
\fi
#2}}
\providecommand{\BIBdecl}{\relax}
\BIBdecl

\bibitem{Alur90}
R.~Alur and D.~Dill, ``Automata for modeling real-time systems,'' in
  \emph{Automata, Languages and Programming}.\hskip 1em plus 0.5em minus
  0.4em\relax Springer Berlin Heidelberg, 1990, pp. 322--335.

\bibitem{alur1994a}
R.~Alur and D.~L. Dill, ``A theory of timed automata,'' \emph{Theoretical
  Computer Science}, vol. 126, no.~2, pp. 183--235, 1994.

\bibitem{alur1993parametric}
R.~Alur, T.~A. Henzinger, and M.~Y. Vardi, ``Parametric real-time reasoning,''
  in \emph{Proceedings of the twenty-fifth annual ACM symposium on Theory of
  computing}.\hskip 1em plus 0.5em minus 0.4em\relax ACM, 1993, pp. 592--601.

\bibitem{annichini2000symbolic}
A.~Annichini, E.~Asarin, and A.~Bouajjani, ``Symbolic techniques for parametric
  reasoning about counter and clock systems,'' in \emph{Computer Aided
  Verification}.\hskip 1em plus 0.5em minus 0.4em\relax Springer, 2000, pp.
  419--434.

\bibitem{bandini2001application}
G.~Bandini, R.~Spelberg, R.~C. de~Rooij, and W.~Toetenel, ``Application of
  parametric model checking-the root contention protocol,'' in \emph{System
  Sciences, 2001. Proceedings of the 34th Annual Hawaii International
  Conference on}.\hskip 1em plus 0.5em minus 0.4em\relax IEEE, 2001, pp.
  10--pp.

\bibitem{HUNE2002183}
T.~Hune, J.~Romijn, M.~Stoelinga, and F.~Vaandrager, ``Linear parametric model
  checking of timed automata,'' \emph{The Journal of Logic and Algebraic
  Programming}, vol. 52-53, pp. 183 -- 220, 2002.

\bibitem{bozzelli2009decision}
L.~Bozzelli and S.~La~Torre, ``Decision problems for lower/upper bound
  parametric timed automata,'' \emph{Formal Methods in System Design}, vol.~35,
  no.~2, p. 121, 2009.

\bibitem{alur2001parametric}
R.~Alur, K.~Etessami, S.~La~Torre, and D.~Peled, ``Parametric temporal logic
  for “model measuring”,'' \emph{ACM Transactions on Computational Logic
  (TOCL)}, vol.~2, no.~3, pp. 388--407, 2001.

\bibitem{KP12a}
M.~Knapik and W.~Penczek, ``Smt-based parameter synthesis for l/u automata,''
  in \emph{PNSE}, 2012, pp. 77--92.

\bibitem{AL17}
{\'E}.~Andr{\'e} and D.~Lime, ``Liveness in l/u-parametric timed automata,'' in
  \emph{ACSD}, 2017, pp. 9--18.

\bibitem{jovanovic2015integer}
A.~Jovanovi{\'c}, D.~Lime, and O.~H. Roux, ``Integer parameter synthesis for
  real-time systems,'' \emph{IEEE Transactions on Software Engineering},
  vol.~41, no.~5, pp. 445--461, 2015.

\bibitem{bundala2014advances}
D.~Bundala and J.~Ouaknine, ``Advances in parametric real-time reasoning,'' in
  \emph{International Symposium on Mathematical Foundations of Computer
  Science}.\hskip 1em plus 0.5em minus 0.4em\relax Springer, 2014, pp.
  123--134.

\bibitem{frehse2008counterexample}
G.~Frehse, S.~K. Jha, and B.~H. Krogh, ``A counterexample-guided approach to
  parameter synthesis for linear hybrid automata,'' in \emph{Hybrid Systems:
  Computation and Control}.\hskip 1em plus 0.5em minus 0.4em\relax Springer,
  2008, pp. 187--200.

\bibitem{andre2009inverse}
{\'E}.~Andr{\'e}, T.~Chatain, L.~Fribourg, and E.~Encrenaz, ``An inverse method
  for parametric timed automata,'' \emph{International Journal of Foundations
  of Computer Science}, vol.~20, no.~05, pp. 819--836, 2009.

\bibitem{andre2015language}
{\'E}.~Andr{\'e} and N.~Markey, ``Language preservation problems in parametric
  timed automata,'' in \emph{International Conference on Formal Modeling and
  Analysis of Timed Systems}.\hskip 1em plus 0.5em minus 0.4em\relax Springer,
  2015, pp. 27--43.

\bibitem{Knapik2012BMC}
M.~Knapik and W.~Penczek, ``Transactions on petri nets and other models of
  concurrency v,'' K.~Jensen, S.~Donatelli, and J.~Kleijn, Eds.\hskip 1em plus
  0.5em minus 0.4em\relax Berlin, Heidelberg: Springer-Verlag, 2012, ch.
  Bounded Model Checking for Parametric Timed Automata, pp. 141--159.

\bibitem{benes2015language}
N.~Benes, P.~Bezdĕk, K.~G. Larsen, and J.~Srba, ``Language emptiness of
  continuous-time parametric timed automata,'' \emph{international colloquium
  on automata languages and programming}, pp. 69--81, 2015.

\bibitem{LSGA17}
J.~Li, J.~Sun, B.~Gao, and A.~{\'E}tienne, ``Classification-based parameter
  synthesis for parametric timed automata,'' in \emph{Formal Methods and
  Software Engineering}.\hskip 1em plus 0.5em minus 0.4em\relax Springer
  International Publishing, 2017, pp. 243--261.

\bibitem{Andr16}
{\'E}.~Andr{\'e}, ``What's decidable about parametric timed automata?'' in
  \emph{Formal Techniques for Safety-Critical Systems}.\hskip 1em plus 0.5em
  minus 0.4em\relax Springer International Publishing, 2016, pp. 52--68.

\bibitem{alur1999timed}
R.~Alur, ``Timed automata,'' in \emph{Computer Aided Verification}.\hskip 1em
  plus 0.5em minus 0.4em\relax Springer, 1999, pp. 8--22.

\bibitem{bengtsson2004timed}
J.~Bengtsson and W.~Yi, ``Timed automata: Semantics, algorithms and tools,'' in
  \emph{Lectures on Concurrency and Petri Nets}.\hskip 1em plus 0.5em minus
  0.4em\relax Springer, 2004, pp. 87--124.

\bibitem{sharma2012interpolants}
R.~Sharma, A.~V. Nori, and A.~Aiken, ``Interpolants as classifiers,'' in
  \emph{Computer Aided Verification}.\hskip 1em plus 0.5em minus 0.4em\relax
  Springer, 2012, pp. 71--87.

\end{thebibliography}

\end{document}